\documentclass{article}%
\usepackage{amsmath}
\usepackage{amssymb}
\usepackage{amsfonts}
\usepackage{graphicx}%
\providecommand{\U}[1]{\protect\rule{.1in}{.1in}}
\newtheorem{theorem}{Theorem}

\newtheorem{conjecture}[theorem]{Conjecture}
\newtheorem{corollary}[theorem]{Corollary}

\newtheorem{example}[theorem]{Example}

\newtheorem{proposition}[theorem]{Proposition}

\newenvironment{proof}[1][Proof]{\noindent\textbf{#1.} }{\ \rule{0.5em}{0.5em}}
\begin{document}

\title{On metric of quantum channel spaces}
\author{Keiji Matsumoto\\National Institute of Informatics, Tokyo, Japan}
\maketitle

\begin{abstract}
So far, there have been plenty of literatures on the metric in the space of
probability distributions and quantum states. As for channels, however, only a
little had been known. In this paper, we impose monotonicity by concatenation
of channels before and after the given channel families, and invariance by
tensoring identity channels. Under these axioms, we identify the largest and
the smallest metrics. Also, we studied asymptotic theory of metric in parallel
and adaptive repetition settings, and applied them to the study of channel
estimation. First we express the achievable lower bound of the mean square
error (MSE) of an \ estimate by a monotone channel metric, and show this
equals $O\left(  1/n\right)  $ for noisy channels, \ where $n$ is the number
of times of channel use. This result shows Heisenberg rate, or $\,O\left(
1/n^{2}\right)  $-rate of the MSE observed in case of estimation of unitary,
collapses with very small arbitrary noise.

\end{abstract}

\section{Introduction}

The aim of the manuscript is to characterize monotone (not necessarily
Riemannian) metric in the space of quantum channels, or CPTP maps, and
application of the theory to channel estimation problem.

So far, there have been plenty of literatures on the metric in the space of
probability distributions and quantum states. Cencov, sometime in 1970s,
proved the monotone metric in probability distribution space is unique up to
constant multiple, and identical to Fisher information metric \cite{Cencov}.
He also discussed invariant connections in the same space. Amari and others
independently worked on the same objects, especially from differential
geometrical view points, and applied to number of problems in mathematical
statistics, learning theory, time series analysis, dynamic systems, control
theory, and so on\cite{Amari}\cite{AmariNagaoka}. Quantum mechanical states
are discussed in literatures such as \cite{AmariNagaoka}\cite{Matsumoto}%
\cite{Matsumoto}\cite{Petz}. Among them Petz\thinspace\cite{Petz}
characterized all the monotone metrics in the quantum state space using
operator mean.

As for channels, however, only a little had been known. To my knowledge, there
had been no study about axiomatic characterization of distance measures in the
classical or quantum channel space.

In this paper, we impose monotonicity by concatenation of channels before and
after the given channel families, and invariance by tensoring identity
channels. Notably, we do \textit{not }suppose a metric is Riemannian, since,
as was shown in , this assumption is not compatible with other assumptions.

Under these axioms, we identify the largest and the smallest metrics. Also, we
studied asymptotic theory of metric in parallel and adaptive repetition
settings, and applied them to the study of channel estimation. First we
express the achievable lower bound of the mean square error (MSE) of an
\ estimate by a monotone channel metric, and show this equals $O\left(
1/n\right)  $ for noisy channels, \ where $n$ is the number of times of
channel use. This result shows Heisenberg rate, or $\,O\left(  1/n^{2}\right)
$-rate of the MSE observed in case of estimation of unitary, collapses with
very small arbitrary noise.

\section{Notations and conventions}

\begin{itemize}
\item $\mathcal{H}_{\mathrm{in}}$ ($\mathcal{H}_{\mathrm{out}}$) :the Hilbert
space for the input (output)

\item $\mathcal{S}_{\mathrm{in}}$ ($\mathcal{S}_{\mathrm{out}}$) : the
totality of the quantum states living in $\mathcal{H}_{\mathrm{in}}$
($\mathcal{H}_{\mathrm{out}}$). In this paper, the existence of density
operator is always assumed. Hence, $\mathcal{S}_{\mathrm{in}}$ ($\mathcal{S}%
_{\mathrm{out}}$) is equivalent to the totality of density operators.

\item $\mathcal{S}\left(  \mathcal{H}\right)  $ : the totality of the quantum
states living in $\mathcal{H}$.

\item $\mathcal{QC}$: the totality of channels which sends an element of
$\mathcal{S}_{\mathrm{in}}$ to an element of $\mathcal{S}_{\mathrm{out}}$

\item $\mathcal{QC}\left(  \mathcal{S}_{1},\mathcal{S}_{2}\right)  $ : the
totality of channels which sends an element of $\mathcal{S}_{1}$ to an element
of $\mathcal{S}_{2}$. Abbreviated form $\mathcal{QC}$ indicates that $\left(
\mathcal{S}_{1},\mathcal{S}_{2}\right)  =\left(  \mathcal{S}_{\mathrm{in}%
},\mathcal{S}_{\mathrm{out}}\right)  $. Also, $\mathcal{QC}\left(
\mathcal{S}\left(  \mathcal{H}_{1}\right)  ,\mathcal{S}\left(  \mathcal{H}%
_{2}\right)  \right)  $ is abbreviated as $\mathcal{QC}\left(  \mathcal{H}%
_{1},\mathcal{H}_{2}\right)  $. Also, $\mathcal{QC}\left(  \mathcal{S}\right)
$ and $\mathcal{QC}\left(  \mathcal{H}\right)  $ means $\mathcal{QC}\left(
\mathcal{S},\mathcal{S}\right)  $ and $\mathcal{QC}\left(  \mathcal{H}%
,\mathcal{H}\right)  $, respectively.

\item A quantum state $\rho$ is identified with the channel which sends all
the input states to $\rho$.

\item $\mathcal{T}_{\cdot}\left(  \cdot\right)  $: tangent space

\item $\delta$ etc. : an element of $\mathcal{T}_{\rho}\left(  \mathcal{S}%
_{\mathrm{in}}\right)  $ , $\mathcal{T}_{p}\left(  \mathcal{P}_{\Omega
}\right)  $ , etc.

\item $\Delta$ etc. : an element of $\mathcal{T}_{\Phi}\left(  QC\right)  $

\item An element $\delta$ of $\mathcal{T}_{\rho}\left(  \mathcal{S}\right)  $
etc. is identified with an element of $\tau c\left(  \mathcal{S}_{\mathrm{in}%
}\right)  $ such that $\mathrm{tr}\delta=0$.

\item $g_{\rho}\left(  \delta\right)  $: \ square of a norm in $\mathcal{T}%
_{\rho}\left(  \mathcal{S}\right)  $

\item $h_{\rho}\left(  \delta\right)  $: square of a norm in $\mathcal{T}%
_{p}\left(  \mathcal{P}\right)  $

\item $G_{\Phi}\left(  \Delta\right)  $: square of a norm in $\mathcal{T}%
_{\Phi}(\mathcal{QC})$

\item $J_{p}\left(  \delta\right)  $ : classical Fisher information

\item $J_{p}^{S}\left(  \delta\right)  $ : SLD Fisher information. $J_{p}%
^{S}\left(  \delta\right)  :=\mathrm{tr}\,\rho\left(  L_{\rho}^{S}\right)
^{2}$, where $L_{\rho}^{S}$ is symmetric logarithmic derivative (SLD), or

the solution to the equation $\delta=\frac{1}{2}\left(  L_{\rho}^{S}\rho+\rho
L_{\rho}^{S}\right)  $.

\item $J_{p}^{S}\left(  \delta\right)  $ : RLD Fisher information. $J_{p}%
^{S}\left(  \delta\right)  :=\Re\mathrm{tr}\,\rho\left(  L_{\rho}^{R}\right)
^{\dagger}L_{\rho}^{R}$, where $L_{\rho}^{R}$ is right logarithmic derivative
(RLD), or

the solution to the equation $\delta=L_{\rho}^{R}\rho$.

\item The local data at $p$, etc.: the pair $\left\{  p,\delta\right\}  $, etc.

\item $\Phi\left(  \cdot|x\right)  \in\mathcal{P}_{\mathrm{out}}$ : the
distribution of the output alphabet when the input is $x$

\item $\Delta\left(  \cdot|x\right)  \in\mathcal{T}_{p}\left(  \mathcal{P}%
_{\mathrm{out}}\right)  $ is defined as the infinitesimal increment of above

\item $\mathbf{I}$: identity

\item $\delta_{e}:=\frac{\delta}{p}$

\item $\delta^{\left(  n\right)  }:=\delta\otimes p^{\otimes n-1}%
+p\otimes\delta\otimes p^{\otimes n-2}+\cdots+p^{\otimes n-1}\otimes\delta
\in\mathcal{T}_{p}\left(  \mathcal{P}^{\otimes n}\right)  $

\item $\delta_{e}^{\left(  n\right)  }:=\delta_{e}\otimes1^{\otimes
n-1}+1\otimes\delta_{e}\otimes1^{\otimes n-2}+\cdots+1^{\otimes n-1}%
\otimes\delta_{e}$

\item $\Delta^{\left(  n\right)  }:=\Delta\otimes\Phi^{\otimes n-1}%
+\Phi\otimes\Delta\otimes\Phi^{\otimes n-2}+\cdots+\Phi^{\otimes n-1}%
\otimes\Delta\in\mathcal{T}_{\Phi}\left(  \mathcal{C}^{\otimes n}\right)  $

\item $\mathrm{N}\left(  a,\sigma^{2}\right)  $ : Gaussian with mean $a$ and
the variance $\sigma^{2}$.

\item $\delta\mathrm{N}\left(  a,\sigma^{2}\right)  $ is singed measure
defined by $\delta\mathrm{N}\left(  a,\sigma^{2}\right)  \left(  B\right)
=\frac{1}{\sqrt{2\pi}\sigma}\int_{B}\frac{x-a}{\sigma^{2}}\exp\left[
-\frac{1}{2\sigma^{2}}\left(  x-a\right)  ^{2}\right]  \mathrm{d}x$. This
corresponds to the tangent vector of Gaussian shift family with the variance
$\sigma^{2}$.
\end{itemize}

\section{Single-copy theory}

\subsection{Axioms}

\begin{description}
\item[(M)] $G_{\Phi}\left(  \Delta\right)  \geq G_{\Phi\circ\Psi}\left(
\Delta\circ\Psi\right)  $, $G_{\Phi}\left(  \Delta\right)  \geq G_{\Psi
\circ\Phi}\left(  \Psi\circ\Delta\right)  $

\item[(E)] $G_{\Phi\otimes\mathbf{I}}\left(  \Delta\otimes\mathbf{I}\right)
=G_{\Phi}\left(  \Delta\right)  $

\item[(N)] $G_{p}\left(  \delta\right)  =J_{p}\left(  \delta\right)  $
\end{description}

\subsection{Estimation of channel and $G^{\min}$}

Consider estimation of an unknown channel, which is drawn from the family
$\left\{  \Phi_{\theta}\right\}  _{\theta\in%
\mathbb{R}
}$, where $\theta\in%
\mathbb{R}
$ is unknown scalar parameter. The asymptotic mean-square error of probability
distribution and quantum state is inversely proportional to Fisher information
$J_{p}\left(  \delta\right)  $ and SLD Fisher information $J_{\rho}^{S}\left(
\delta\right)  $, respectively. Hence, it is natural to consider
\[
G_{\Phi}^{\min}\left(  \Delta\right)  :=\sup_{\substack{\rho\in\mathcal{S}%
\left(  \mathcal{H}_{\mathrm{n}}\otimes\mathcal{K}\right)  \\M\in
\mathcal{M}_{\mathrm{out}}}}J_{M\circ\left(  \Phi\otimes\mathbf{I}\right)
\left(  \rho\right)  }\left(  M\circ\left(  \Delta\otimes\mathbf{I}\right)
\left(  \rho\right)  \right)  =\sup_{\rho\in\mathcal{S}\left(  \mathcal{H}%
_{\mathrm{n}}\otimes\mathcal{K}\right)  }J_{\left(  \Phi\otimes\mathbf{I}%
\right)  \left(  \rho\right)  }^{S}\left(  \left(  \Delta\otimes
\mathbf{I}\right)  \left(  \rho\right)  \right)  ,
\]
where the identity is due to characterization of SLD Fisher information in
\cite{Nagaoka}:

\begin{theorem}
\cite{Nagaoka}
\[
J_{\rho}^{S}\left(  \delta\right)  =\sup_{M}J_{M\left(  \rho\right)  }\left(
M\left(  \delta\right)  \right)  .
\]

\end{theorem}

\begin{theorem}
\label{th:G>Gmin}Suppose (M) and (N) hold. Then, $\ $%
\[
G_{\Phi}\left(  \Delta\right)  \geq G_{\Phi}^{\min}\left(  \Delta\right)
\]
Also, $G_{\Phi}^{\min}\left(  \Delta\right)  $ satisfies (M), (E), and (N).
\end{theorem}

\begin{proof}%
\[
G_{\Phi}\left(  \Delta\right)  =G_{\Phi\otimes\mathbf{I}}\left(  \Delta
\otimes\mathbf{I}\right)  \geq G_{M\circ\left(  \Phi\otimes\mathbf{I}\right)
\left(  \rho\right)  }\left(  M\circ\left(  \Delta\otimes\mathbf{I}\right)
\left(  \rho\right)  \right)  =J_{M\circ\left(  \Phi\otimes\mathbf{I}\right)
\left(  \rho\right)  }\left(  M\circ\left(  \Delta\otimes\mathbf{I}\right)
\left(  \rho\right)  \right)  .
\]
Hence, we have inequality. That $G_{\Phi}^{\min}\left(  \Delta\right)  $
satisfies (M1), (M2), (E), and (N) is trivial.
\end{proof}

\subsection{Tangent simulation of channel family and $G^{\max}$}

Suppose we have to fabricate a channel $\Phi_{\theta}$, which is drawn from a
family $\left\{  \Phi_{\theta}\right\}  $, without knowing the value of
$\theta$ but with a probability distribution $q_{\theta\text{ }}$ or
$\rho_{\theta}$ drawn from a family $\left\{  q_{\theta}\right\}  $ or
$\left\{  \rho_{\theta}\right\}  $. More specifically, we need a channel
$\Lambda$ with
\begin{equation}
\Phi_{\theta}=\Lambda\circ\left(  \mathbf{I}\otimes q_{\theta}\right)  ,
\label{simulation-1}%
\end{equation}
or%
\begin{equation}
\Phi_{\theta}=\Lambda\circ\left(  \mathbf{I}\otimes\sigma_{\theta}\right)  ,\,
\label{simulation-3}%
\end{equation}
Here, note that $\Lambda$ should not vary with the parameter $\theta$. Note
also that the former is a special case of the latter. Also, giving the value
of $\theta$ with infinite precision corresponds to the case of having the
delta distribution peaked at $\theta$. This is channel version of
randomization criteria for deficiency, which is a fundamental concept in
statistical decision theory \cite{Torgersen}.

Differentiating the both ends of (\ref{simulation-1}), (\ref{simulation-2}),
and (\ref{simulation-3}), we obtain%

\begin{equation}
\Delta=\Lambda\circ\left(  \mathbf{I}\otimes\delta\right)  ,\text{ }
\label{simulation-tangent-1}%
\end{equation}
where $\Delta\in\mathcal{T}_{\Phi}\left(  \mathcal{C}\right)  $, $\delta
\in\mathcal{T}_{q}\left(  \mathcal{P}_{\mathrm{pr}}\right)  \,$(, or
$\mathcal{T}_{\rho}\left(  \mathcal{S}_{\mathrm{pr}}\right)  $).

In the manuscript, we consider \textit{classical} \textit{tangent simulation
(, or quantum tangent simulation)}, or the triplet $\left\{  q,\delta
,\Lambda\right\}  $ (, or $\left\{  \sigma,\delta,\Lambda\right\}  $ )
satisfying (\ref{simulation-1}) (, or (\ref{simulation-3}) ) and
(\ref{simulation-tangent-1}), at the point $\Phi_{\theta}=\Phi$ only. Note
that classical and quantum tangent simulation of $\left\{  \Phi,\Delta
\right\}  $ is equivalent to simulation of the channel family $\left\{
\Phi_{\theta+t}=\Phi+t\Delta\right\}  _{t}$. (This is channel analogue of
local deficiency in statistical decision theory \cite{Torgersen}.)

Based on tangent simulation, we define :%

\begin{align*}
G_{\Phi}^{\max}\left(  \Delta\right)   &  :=\inf\left\{  J_{q}\left(
\delta\right)  \,;\{\,\Lambda,\,q,\delta\}\text{ is a classical tangent
simulation of }\left\{  q,\delta\right\}  \text{ }\right\}  ,\\
&  =\inf\left\{  J_{\sigma}^{R}\left(  \delta\right)  \,;\{\,\Lambda
,\,\sigma,\delta\}\text{ is a quantum tangent simulation of }\left\{
q,\delta\right\}  \text{ }\right\}  ,
\end{align*}
where the identity in the second line is due to characterization of RLD in
\cite{Matsumoto:rld}.

\begin{theorem}
\label{th:G<Gmax}Suppose (M), (E) and (N) hold. Then
\[
G_{\Phi}\left(  \Delta\right)  \leq G_{\Phi}^{\max}\left(  \Delta\right)  .
\]
Also, $G_{\Phi}^{\max}\left(  \Delta\right)  $ satisfies (M), (E), and (N).
\end{theorem}

\begin{proof}%
\[
J_{q}\left(  \delta\right)  =G_{q}\left(  \delta\right)  =G_{\mathbf{I}\otimes
q}\left(  \mathbf{I}\otimes\delta\right)  \geq G_{\Lambda\circ\left(
\mathbf{I}\otimes q\right)  }\left(  \Lambda\circ\left(  \mathbf{I}%
\otimes\delta\right)  \right)  =G_{\Phi}\left(  \Delta\right)  .
\]
So we have the inequality. That $G_{\Phi}^{\max}\left(  \Delta\right)  $
satisfies (M), (E), and (N) is trivial.
\end{proof}

\begin{corollary}
\label{cor:gmax>gmin}%
\[
G_{\Phi}^{\max}\left(  \Delta\right)  \geq G_{\Phi}^{\min}\left(
\Delta\right)  .
\]

\end{corollary}

\begin{example}
\cite{Fujiwara}Consider the following family of channels :%
\[
\Lambda_{\theta}\left(  \rho\right)  =\left(  1-p_{x}-p_{y}-p_{z}\right)
\rho+p_{x}X\rho X+p_{y}Y\rho Y+p_{z}Z\rho Z,
\]
where $X$,$Y$,$Z$ are Pauli matrices and $p_{x}$,$p_{y}$, $p_{z}$ are scalar
functions of $\theta$. In other words, consider random application of Pauli
matrices with unknown probability distribution $p_{\theta}=\left(
p_{x}\left(  \theta\right)  ,p_{y}\left(  \theta\right)  ,p_{z}\left(
\theta\right)  \right)  $. \ Therefore,
\[
G_{\Lambda_{\theta}}^{\max}\left(  \Delta_{\theta}\right)  \leq J_{p_{\theta}%
}\left(  \delta_{\theta}\right)
\]
where $\Delta_{\theta}=\mathrm{d}\Lambda_{\theta}/\mathrm{d}\theta$ and
$\delta_{\theta}=\mathrm{d}p_{\theta}/\mathrm{d}\theta$. On the other hand,
let
\begin{align*}
\left\vert \mathrm{Bell}_{1}\right\rangle  &  :=\frac{1}{\sqrt{2}}\left(
\left\vert 00\right\rangle +\left\vert 11\right\rangle \right) \\
\left\vert \mathrm{Bell}_{2}\right\rangle  &  :=I\otimes X\left\vert
\mathrm{Bell}_{1}\right\rangle =\frac{1}{\sqrt{2}}\left(  \left\vert
01\right\rangle +\left\vert 10\right\rangle \right) \\
\left\vert \mathrm{Bell}_{3}\right\rangle  &  :=\sqrt{-1}\left(  I\otimes
Y\right)  \left\vert \mathrm{Bell}_{1}\right\rangle =\frac{1}{\sqrt{2}}\left(
\left\vert 01\right\rangle -\left\vert 10\right\rangle \right) \\
\left\vert \mathrm{Bell}_{4}\right\rangle  &  :=\left(  I\otimes Z\right)
\left\vert \mathrm{Bell}_{1}\right\rangle =\frac{1}{\sqrt{2}}\left(
\left\vert 00\right\rangle -\left\vert 11\right\rangle \right)  .
\end{align*}
Observe that they are orthogonal with each other. Hence, as Fujiwara and
others had pointed out, by inserting one part of $\left\vert \mathrm{Bell}%
_{1}\right\rangle $ and measuring the output, we can identify which Pauli
matrix was multiplied. Therefore,
\[
G_{\Lambda_{\theta}}^{\min}\left(  \Delta_{\theta}\right)  \geq J_{p_{\theta}%
}\left(  \delta_{\theta}\right)  .
\]
Hence, after all,
\[
G_{\Lambda_{\theta}}^{\min}\left(  \Delta_{\theta}\right)  =G_{\Lambda
_{\theta}}^{\max}\left(  \Delta_{\theta}\right)  =g_{p_{\theta}}\left(
\delta_{\theta}\right)  \text{.}%
\]

\end{example}

\subsection{Quantum states}

A quantum state can be viewed as a quantum channel with constant output. In
\cite{Matsumoto:rld}, (M2) and (N) implies that
\[
J_{\rho}^{S}\left(  \delta\right)  \leq G_{\rho}\left(  \delta\right)  \leq
J_{\rho}^{R}\left(  \delta\right)  .
\]

\section{Asymptotic theory : parallel version}

\textit{Parallel use} of $n$ of $\Phi$ means that we are given $\Phi^{\otimes
n}$, send in a big input $\rho\in\mathcal{S}_{\mathrm{in}}^{n}$ to
$\Phi^{\otimes n}$.

\subsection{Additional axioms: parallel version}

\begin{description}
\item[(Ap)] \textit{(parallel asymptotic weak additivity)} $\lim
_{n\rightarrow\infty}\frac{1}{n}G_{\Phi^{\otimes n}}\left(  \Delta^{\left(
n\right)  }\right)  =G_{\Phi}\left(  \Delta\right)  $
\end{description}

\subsection{ $G^{p,\min}$ and $G^{p,\max}$}

We define
\begin{align*}
G_{\Phi}^{\min,p}\left(  \Delta\right)   &  :=\lim_{n\rightarrow\infty}%
\frac{1}{n}G_{\Phi^{\otimes n}}^{\min}\left(  \Delta^{\left(  n\right)
}\right)  ,\\
G_{\Phi}^{\max,p}\left(  \Delta\right)   &  :=\lim_{n\rightarrow\infty}%
\frac{1}{n}G_{\Phi^{\otimes n}}^{\max}\left(  \Delta^{\left(  n\right)
}\right)  ,
\end{align*}

\begin{theorem}
(M), (E), (Ap), and (N) implies that
\[
G_{\Phi}^{\min,p}\left(  \Delta\right)  \leq G_{\Phi}\left(  \Delta\right)
\leq G_{\Phi}^{\max,p}\left(  \Delta\right)  .
\]
Also, $G_{\Phi}^{\min,p}\left(  \Delta\right)  $ $\ $and $G_{\Phi}^{\max
,p}\left(  \Delta\right)  $ satisfy (M), (E), (Ap), and (N).
\end{theorem}

\begin{proof}
By Theorem\thinspace\ref{th:G>Gmin} and Theorem\thinspace\ref{th:G<Gmax},
\[
\frac{1}{n}G_{\Phi^{\otimes n}}^{\min}\left(  \Delta^{\left(  n\right)
}\right)  \leq\frac{1}{n}G_{\Phi^{\otimes n}}\left(  \Delta^{\left(  n\right)
}\right)  \leq\frac{1}{n}G_{\Phi^{\otimes n}}^{\max}\left(  \Delta^{\left(
n\right)  }\right)  \,.
\]
Taking sup of the last end and letting $n\rightarrow\infty$, we have the
assertion That $G_{\Phi}^{\min,p}\left(  \Delta\right)  $ and $G_{\Phi}%
^{\max,p}\left(  \Delta\right)  $ satisfy (M), (E), (Ap), and (N) is trivial.
\end{proof}

\section{Asymptotic theory : adaptive version}

\subsection{Adaptive repetition}

In estimating channel, we may use it sequentially, applying some channel
$\Psi_{\kappa}$ between $k$th and( $k-1)$th application of $\Phi
\otimes\mathbf{I}$:%

\begin{equation}
\prod_{k=n}^{1}\left\{  \left(  \Phi\otimes\mathbf{I}\right)  \circ
\Psi_{\kappa}\right\}  . \label{seq-prod}%
\end{equation}
To indicate such use, we define $n$-adaptive repetition of $\Phi$ by
$\Phi^{\#n}$. Formal definition is that $\Phi^{\#n}$ is a linear map which
sends the pair $\mathbf{\Psi}^{n}:=\left(  \Psi_{1},\Psi_{2},\cdots,\Psi
_{n}\right)  $ to (\ref{seq-prod}). We also define
\[
\Delta^{\left(  \#n\right)  }:=\Delta\#\Phi^{\#n-1}+\Phi\#\Delta\#\Phi
^{\#n-2}+\cdots+\Phi^{\#n-1}\#\Delta.
\]
Here, $\Delta$ is identified with a linear map from operators to operators.
For the sake of briefness, we denote:%
\begin{align}
\Phi^{\#n}\left(  \mathbf{\Psi}^{n}\right)   &  :=\prod_{k=n}^{1}\left\{
\left(  \Phi\otimes\mathbf{I}\right)  \circ\Psi_{\kappa}\right\}
,\,\nonumber\\
\Delta^{\left(  \#n\right)  }\left(  \mathbf{\Psi}^{n}\right)   &  :=\left\{
\left(  \Delta\otimes\mathbf{I}\right)  \circ\Psi_{n}\right\}  \circ
\prod_{k=n-1}^{1}\left\{  \left(  \Phi\otimes\mathbf{I}\right)  \circ
\Psi_{\kappa}\right\} \nonumber\\
&  +\left\{  \left(  \Phi\otimes\mathbf{I}\right)  \circ\Psi_{n}\right\}
\circ\left\{  \left(  \Delta\otimes\mathbf{I}\right)  \circ\Psi_{n-1}\right\}
\circ\prod_{k=n-2}^{1}\left\{  \left(  \Phi\otimes\mathbf{I}\right)  \circ
\Psi_{\kappa}\right\}  +\cdots\nonumber\\
&  +\prod_{k=n}^{2}\left\{  \left(  \Phi\otimes\mathbf{I}\right)  \circ
\Psi_{\kappa}\right\}  \circ\left\{  \left(  \Delta\otimes\mathbf{I}\right)
\circ\Psi_{1}\right\}  .
\end{align}

One can define%

\begin{align*}
G_{\Delta^{\#n}}^{\min}\left(  \Phi^{\#n}\right)   &  :=\sup\left\{
J_{\tilde{p}}\left(  \tilde{\delta}\right)  \,;\tilde{p}=M\circ\Phi
^{\#n}\left(  \mathbf{\Psi}^{n}\right)  \left(  \rho\right)  ,\,\tilde{\delta
}=M\circ\Delta^{\left(  \#n\right)  }\left(  \mathbf{\Psi}^{n}\right)  \left(
\rho\right)  \right\}  ,\\
G_{\Delta^{\#n}}^{\max}\left(  \Phi^{\#n}\right)   &  :=\inf J_{q}\left(
\delta^{\prime}\right)  ,
\end{align*}
where the infimum in the second definition is taken over all $\left\{
q,\delta^{\prime}\right\}  $ which satisfies for some $\mathbf{\Lambda}%
^{n}:=\left(  \Lambda_{1},\Lambda_{2},\cdots,\Lambda_{n}\right)  $
\begin{align}
\prod_{k=n}^{1}\left\{  \left(  \Lambda_{\kappa}\otimes\mathbf{I}%
\otimes\mathbf{I}\right)  \circ\Psi_{\kappa}\otimes\mathbf{I}\right\}  \left(
\rho\otimes q\right)   &  =\Phi^{\#n}\left(  \mathbf{\Psi}^{n}\right)  \left(
\rho\right)  \,\label{simulation-tangent-1-a}\\
\prod_{k=n}^{1}\left\{  \left(  \Lambda_{\kappa}\otimes\mathbf{I}%
\otimes\mathbf{I}\right)  \circ\left(  \Psi_{\kappa}\otimes\mathbf{I}\right)
\right\}  \left(  \rho\otimes\delta\right)   &  =\Delta^{\left(  \#n\right)
}\left(  \mathbf{\Psi}^{n}\right)  \left(  \rho\right)  .
\label{asym-simulation-tangent-2-a}%
\end{align}
\textit{Classical tangent simulation} of $\left\{  \Phi^{\#n},\Delta^{\left(
\#n\right)  }\right\}  $ is defined as a pair $\left\{  \mathbf{\Lambda}%
^{n},q,\delta^{\prime}\right\}  $ with (\ref{simulation-tangent-1-a}),
(\ref{simulation-tangent-2-a}).

\subsection{$G^{\min,a}$ and $G^{\max,a}$}

We define
\begin{align*}
G_{\Phi}^{\min,a}\left(  \Delta\right)   &  :=\lim_{n\rightarrow\infty}%
\frac{1}{n}G_{\Phi^{\#n}}^{\min}\left(  \Delta^{\left(  \#n\right)  }\right)
,\\
G_{\Phi}^{\max,a}\left(  \Delta\right)   &  :=\lim_{n\rightarrow\infty}%
\frac{1}{n}G_{\Phi^{\#n}}^{\max}\left(  \Delta^{\left(  \#n\right)  }\right)
.
\end{align*}

Then we have the following theorems.

\begin{theorem}%
\[
G_{\Phi}^{\min,a}\left(  \Delta\right)  \leq G_{\Phi}^{\max,a}\left(
\Delta\right)  .
\]
Also, $G_{\Phi}^{\min,a}\left(  \Delta\right)  $ and $G_{\Phi}^{\max,a}\left(
\Delta\right)  $ satisfy (M), (E) and (N) and (Aa).
\end{theorem}

\begin{proof}
That $G_{\Phi}^{\min,a}\left(  \Delta\right)  $ and $G_{\Phi}^{\max,a}\left(
\Delta\right)  $ satisfy (M), (E), and (N) is trivial. So we prove the
inequality. Let%
\[
\tilde{p}^{n}=M\circ\Phi^{\#n}\left(  \mathbf{\Psi}^{n}\right)  \left(
\rho^{n}\right)  \,,\,\tilde{\delta}^{n}=M^{n}\circ\Delta^{\left(  \#n\right)
}\left(  \mathbf{\Psi}^{n}\right)  \left(  \rho^{n}\right)
\]
and let $\left\{  \mathbf{\Lambda}^{n},q^{n},\delta^{\prime n}\right\}  $ be a
tangent simulation of $\left\{  \Phi^{\#n},\Delta^{\left(  \#n\right)
}\right\}  $.Then by monotonicity of Fisher information, we have
\[
J_{q^{n}}\left(  \delta^{\prime n}\right)  \geq J_{\tilde{p}^{n}}\left(
\tilde{\delta}^{n}\right)  .
\]
Hence, taking infimum of the LHS and the maximum of the RHS and letting
$n\rightarrow\infty$, we obtain the second inequality.
\end{proof}

\begin{proposition}%
\begin{align*}
G_{\Phi}^{\min}\left(  \Delta\right)   &  \leq G_{\Phi}^{\min,p}\left(
\Delta\right)  \leq G_{\Phi}^{\min,a}\left(  \Delta\right) \\
&  \leq G_{\Phi}^{\max,p}\left(  \Delta\right)  \leq G_{\Phi}^{\max,a}\left(
\Delta\right)  \leq G_{\Phi}^{\max}\left(  \Delta\right)  .
\end{align*}

\end{proposition}

\begin{proof}
Non-trivial part is \ $G_{\Phi}^{\min,a}\left(  \Delta\right)  \leq G_{\Phi
}^{\max,p}\left(  \Delta\right)  $. To prove this, it suffices to show
$G_{\Phi}^{\min,a}\left(  \Delta\right)  $ satisfies (A1). Consider%
\begin{align*}
\frac{1}{m}G_{\Phi^{\otimes m}}^{\min,a}\left(  \Delta^{\left(  m\right)
}\right)   &  =\lim_{n\rightarrow\infty}\frac{1}{nm}G_{\left(  \Phi^{\otimes
m}\right)  ^{\#n}}^{\min}\left(  \left(  \Delta^{\left(  m\right)  }\right)
^{\#n}\right) \\
&  \leq\lim_{n\rightarrow p}\frac{1}{nm}G_{\Phi^{\#nm}}^{\min}\left(
\Delta^{\left(  \#mn\right)  }\right) \\
&  =G_{\Phi}^{\min,a}\left(  \Delta\right)  .
\end{align*}
Here, the inequality in the second line holds since operations allowed in the
optimization problem used to define the RHS quantity is richer than those used
to define the LHS quantity.

On the other hand, $\left(  \Phi^{\otimes m}\right)  ^{\#n}$ can be thought as
$m$-parallelization of $n$-adaptive sequence. If we restrict $\mathbf{\Psi
}^{n}$, $M$, and $\rho$ so that there is no entanglement nor interaction
between these parallelization, obtained Fisher information becomes in general
smaller. Hence,%
\begin{align*}
\frac{1}{nm}G_{\left(  \Phi^{\otimes m}\right)  ^{\#n}}^{\min}\left(  \left(
\Delta^{\left(  m\right)  }\right)  ^{\left(  \#n\right)  }\right)   &
=\frac{1}{nm}\sup\left\{  J_{\tilde{p}}\left(  \tilde{\delta}\right)
\,;\tilde{p}=M\circ\left(  \Phi^{\otimes m}\right)  ^{\#n}\left(
\mathbf{\Psi}^{n}\right)  \left(  \rho\right)  ,\,\tilde{\delta}=M\circ\left(
\Delta^{\left(  m\right)  }\right)  ^{\left(  \#n\right)  }\left(
\mathbf{\Psi}^{n}\right)  \left(  \rho\right)  \right\} \\
&  \geq\frac{1}{nm}\sup\left\{  J_{\tilde{p}}\left(  \tilde{\delta}\right)
\,;\tilde{p}=\left(  \tilde{M}\circ\left(  \Phi^{\#n}\right)  \left(
\mathbf{\tilde{\Psi}}^{n}\right)  \left(  \tilde{\rho}\right)  \right)
^{\otimes m},\,\tilde{\delta}=\left(  M\circ\Delta^{\left(  \#n\right)
}\left(  \mathbf{\tilde{\Psi}}^{n}\right)  \left(  \tilde{\rho}\right)
\right)  ^{\left(  m\right)  }\right\} \\
&  =\frac{1}{nm}\sup\left\{  J_{\tilde{p}^{\otimes m}}\left(  \tilde{\delta
}^{\otimes m}\right)  \,;\tilde{p}=\tilde{M}\circ\left(  \Phi^{\#n}\right)
\left(  \mathbf{\tilde{\Psi}}^{n}\right)  \left(  \tilde{\rho}\right)
,\,\tilde{\delta}=M\circ\Delta^{\left(  \#n\right)  }\left(  \mathbf{\tilde
{\Psi}}^{n}\right)  \left(  \tilde{\rho}\right)  \right\} \\
&  =\frac{1}{n}\sup\left\{  J_{\tilde{p}}\left(  \tilde{\delta}\right)
\,;\tilde{p}=\tilde{M}\circ\left(  \Phi^{\#n}\right)  \left(  \mathbf{\tilde
{\Psi}}^{n}\right)  \left(  \tilde{\rho}\right)  ,\,\tilde{\delta}%
=M\circ\Delta^{\left(  \#n\right)  }\left(  \mathbf{\tilde{\Psi}}^{n}\right)
\left(  \tilde{\rho}\right)  \right\} \\
&  =\frac{1}{n}G_{\Phi^{\#n}}^{\min}\left(  \Delta^{\left(  \#n\right)
}\right)  .
\end{align*}
By $n\rightarrow\infty$, \ this yields%
\[
\frac{1}{m}G_{\Phi^{\otimes m}}^{\min,a}\left(  \Delta^{\left(  m\right)
}\right)  \geq G_{\Phi}^{\min,a}\left(  \Delta\right)  .
\]
After all, we have
\[
\frac{1}{m}G_{\Phi^{\otimes m}}^{\min,a}\left(  \Delta^{\left(  m\right)
}\right)  =G_{\Phi}^{\min,a}\left(  \Delta\right)  ,
\]
and our assertion is proved.
\end{proof}

\begin{conjecture}
$G_{\Phi}^{\max,a}\left(  \Delta\right)  =G_{\Phi}^{\max}\left(
\Delta\right)  $.
\end{conjecture}

\subsection{Examples}

\subsubsection{Unital qubit channels}

In this case, since $G_{\Phi}^{\max}\left(  \Delta\right)  =G_{\Phi}^{\min
}\left(  \Delta\right)  $, it follows that \
\[
G_{\Phi}^{\min}\left(  \Delta\right)  =G_{\Phi}^{\min,p}\left(  \Delta\right)
=G_{\Phi}^{\min,a}\left(  \Delta\right)  =G_{\Phi}^{\max,p}\left(
\Delta\right)  =G_{\Phi}^{\max,a}\left(  \Delta\right)  =G_{\Phi}^{\max
}\left(  \Delta\right)  .
\]

\subsubsection{QC channels}

If $\left\{  \Phi_{\theta}\right\}  $ is a QC channel, $G_{\Phi}^{\min}\left(
\Delta\right)  =G_{\Phi}^{\min,p}\left(  \Delta\right)  =G_{\Phi}^{\min
,a}\left(  \Delta\right)  $. This is proved as follows. If only classical data
is fed to the succeeding measurement, the fisher information obtained is
$G_{\Phi}^{\min}\left(  \Delta\right)  $, due to. In general, we may have
large input state $\rho_{\mathrm{in}}\in\mathcal{S(}\mathcal{H}_{\mathrm{in}%
}\otimes\mathcal{K})$, where the measurement is applied only to $\mathcal{H}%
_{\mathrm{in}}$, and we are left with the measurement result (classical
information) and the post-measurement state in $\mathcal{K}$. This
post-measurement state is determineded by the measurement data, and therefore
not needed given the measurement result. Since it can be fabricated whenever necessary.

\subsubsection{Quantum states}

A quantum state can be considered as a channel with constant output. Indeed,
if $g$ satisfies (M) and (N),
\[
J_{\rho}^{S}\left(  \delta\right)  \leq g_{\rho}\left(  \delta\right)  \leq
J_{\rho}^{R}\left(  \delta\right)  .
\]
\cite{Matsumoto:rld}. Moreover, it is known that $J^{S}$ and $J^{R}$ satisfy
not only (M) and (N), but also (A1).

\subsubsection{Unitary channels and noisy channels}

If $\left\{  \Phi_{\theta}\right\}  $ are unitary operations, $G_{\Phi}%
^{\min,p}\left(  \Delta\right)  =\infty$. Hence,
\begin{equation}
G_{\Phi}^{\min,p}\left(  \Delta\right)  =G_{\Phi}^{\min,a}\left(
\Delta\right)  =G_{\Phi}^{\max,p}\left(  \Delta\right)  =G_{\Phi}^{\max
,a}\left(  \Delta\right)  =G_{\Phi}^{\max}\left(  \Delta\right)  =\infty.
\label{unitary-infty}%
\end{equation}
If $\Phi$is in the interior of $\mathcal{QC}$, and $\dim\mathcal{H}%
_{\mathrm{in}}<\infty$, $\dim\mathcal{H}_{\mathrm{in}}<\infty$, there is a
$\varepsilon>0$ such that
\[
\Phi+\theta\Delta\in\mathcal{QC},\text{ }\forall\left\vert \theta\right\vert
\leq\varepsilon.
\]
Then, \ $\left\{  \Phi,\Delta\right\}  $ can be simulated by probabilistic
mixture of $\Phi+\varepsilon\Delta$ and $\Phi-\varepsilon\Delta$. More
precisely, let $\Lambda\in\mathcal{QC}\left(  \mathcal{H}_{\mathrm{in}}\otimes%
\mathbb{C}
^{2},\mathcal{H}_{\mathrm{out}}\right)  $ be a channel such that
\begin{align*}
\Lambda\left(  \rho_{\mathrm{in}}\otimes\left\vert 0\right\rangle \left\langle
0\right\vert \right)   &  =\left(  \Phi+\varepsilon\Delta\right)  \left(
\rho_{\mathrm{in}}\otimes\left\vert 0\right\rangle \left\langle 0\right\vert
\right)  ,\\
\Lambda\left(  \rho_{\mathrm{in}}\otimes\left\vert 1\right\rangle \left\langle
1\right\vert \right)   &  =\left(  \Phi-\varepsilon\Delta\right)  \left(
\rho_{\mathrm{in}}\otimes\left\vert 1\right\rangle \left\langle 1\right\vert
\right)  ,
\end{align*}
$q$ is the probability distribution on $\left\{  0,1\right\}  $ with $q\left(
0\right)  =q\left(  1\right)  =\frac{1}{2}$, and $\delta\left(  0\right)
=\left(  2\varepsilon\right)  ^{-1}$, $\delta\left(  1\right)  =-\left(
2\varepsilon\right)  ^{-1}$.

Therefore,
\begin{equation}
G_{\Phi}^{\min,p}\left(  \Delta\right)  \leq G_{\Phi}^{\min,a}\left(
\Delta\right)  \leq G_{\Phi}^{\max,p}\left(  \Delta\right)  \leq G_{\Phi
}^{\max,a}\left(  \Delta\right)  \leq G_{\Phi}^{\max}\left(  \Delta\right)
\leq\left(  \varepsilon\right)  ^{-2}<\infty. \label{noisy-finite}%
\end{equation}

\section{Asymptotic theory of estimation of noisy channels}

\subsection{Cramer-Rao type bound}

An adaptive estimator of the channel family $\left\{  \Phi_{\theta}\right\}  $
is a sequence $\left\{  \rho_{\mathrm{in}}^{n},\mathbf{\Psi}^{n}%
,M^{n}\right\}  _{n=1}^{\infty}\,$of triplet of a pair of channels
$\mathbf{\Psi}^{n}:=\left(  \Psi_{1},\Psi_{2},\cdots,\Psi_{n}\right)  $, the
input state $\rho^{n}\in\mathcal{S}\left(  \mathcal{H}_{\mathrm{in}}%
\otimes\mathcal{K}^{n}\right)  $, and the measurement $M^{n}\in\mathcal{S}%
\left(  \mathcal{M}_{\mathrm{out}}\otimes\mathcal{K}^{n}\right)  $, which
takes values in $%
\mathbb{R}
$. $\left\{  \rho_{\mathrm{in}}^{n},\mathbf{\Psi}^{n},M^{n}\right\}
_{n=1}^{\infty}$ is said to be \textit{asymptotically unbiased} if%

\begin{equation}
\lim_{n\rightarrow\infty}\mathbb{E}_{\theta}\left[  \left\{  \rho
_{\mathrm{in}}^{n},\mathbf{\Psi}^{n},M^{n}\right\}  \right]  =\theta,\quad
\lim_{n\rightarrow\infty}\frac{\mathrm{d}}{\mathrm{d}\theta}\mathbb{E}%
_{\theta}\left[  \left\{  \rho_{\mathrm{in}}^{n},\mathbf{\Psi}^{n}%
,M^{n}\right\}  \right]  =1, \label{asym-unbiased}%
\end{equation}
where $\mathbb{E}_{\theta}\left[  \left\{  \rho_{\mathrm{in}}^{n}%
,\mathbf{\Psi}^{n},M^{n}\right\}  \right]  $ refers to the expectation of
estimate obeying the probability distribution $M^{n}\circ\Phi^{\#n}\left(
\mathbf{\Psi}^{n}\right)  \left(  \rho_{\mathrm{in}}^{n}\right)  $. This is a
regularity condition often imposed on estimators. Given an asymptotically
unbiased estimator, one can define a measurement ' $M_{\theta_{0}}^{n}$ with
measurement result
\begin{align*}
\check{\theta}_{\theta_{0}}^{n}  &  :=\frac{1}{\frac{\mathrm{d}}%
{\mathrm{d}\theta}b_{\theta_{0}}^{n}}\left(  \hat{\theta}^{n}-b_{\theta_{0}%
}^{n}\right)  +\theta_{0},\\
\hat{\theta}^{n}  &  =\left(  \frac{\mathrm{d}}{\mathrm{d}\theta}b_{\theta
_{0}}^{n}\right)  \left(  \check{\theta}_{\theta_{0}}^{n}-\theta_{0}\right)
+b_{\theta_{0}}^{n}%
\end{align*}
where $\hat{\theta}^{n}$ is the measurement result of $M^{n}$ and $b_{\theta
}^{n}:=\mathbb{E}_{\theta}\left[  \left\{  \rho_{\mathrm{in}}^{n}%
,\mathbf{\Psi}^{n},M^{n}\right\}  \right]  $. Then, \ $M_{\theta_{0}}^{n}$
satisfies
\begin{equation}
\mathbb{E}_{\theta_{0}}\left[  \left\{  \rho_{\mathrm{in}}^{n},\mathbf{\Psi
}^{n},M_{\theta_{0}}^{n}\right\}  \right]  =\theta_{0},\,\left.
\frac{\mathrm{d}}{\mathrm{d}\theta}\mathbb{E}_{\theta}\left[  \left\{
\rho_{\mathrm{in}}^{n},\mathbf{\Psi}^{n},M_{\theta_{0}}^{n}\right\}  \right]
\right\vert _{\theta=\theta_{0}}=1, \label{locally-unbiased}%
\end{equation}
and $\theta$%
\begin{align*}
&  \varliminf_{n\rightarrow\infty}\,n\mathbb{E}_{\theta_{0}}\left(
\hat{\theta}^{n}-\theta_{0}\right)  ^{2}=\varliminf_{n\rightarrow\infty
}\,n\mathbb{E}_{\theta_{0}}\left(  \left(  \frac{\mathrm{d}}{\mathrm{d}\theta
}b_{\theta_{0}}^{n}\right)  \left(  \check{\theta}_{\theta_{0}}^{n}-\theta
_{0}\right)  +b_{\theta_{0}}^{n}-\theta_{0}\right)  ^{2}\\
&  =\varliminf_{n\rightarrow\infty}\,n\mathbb{E}_{\theta_{0}}\left(  \left(
\check{\theta}_{\theta_{0}}^{n}-\theta_{0}\right)  +b_{\theta_{0}}^{n}%
-\theta_{0}\right)  ^{2}\\
&  \geq\varliminf_{n\rightarrow\infty}\,n\mathbb{E}_{\theta_{0}}\left(
\check{\theta}_{\theta_{0}}^{n}-\theta_{0}\right)  ^{2}\\
&  \geq\left(  \varlimsup_{n\rightarrow\infty}\frac{1}{n}%
J_{{\normalsize \tilde{p}}^{n}}\left(  \tilde{\delta}^{n}\right)  \right)
^{-1},
\end{align*}
where ${\normalsize \tilde{p}}^{n}:=M\circ\Phi^{\#n}\left(  \mathbf{\Psi}%
^{n}\right)  \left(  \rho_{\mathrm{in}}^{n}\right)  $ and $\tilde{\delta}%
^{n}:=M\circ\Delta^{\left(  \#n\right)  }\left(  \mathbf{\Psi}^{n}\right)
\left(  \rho_{\mathrm{in}}^{n}\right)  $.

Hence, we obtain the \textit{Cramer-Rao type bound} \cite{Matsumoto:akahira}
\begin{equation}
\inf\left\{  \varliminf_{n\rightarrow\infty}\,n\mathbb{E}_{\theta_{0}}\left(
\hat{\theta}^{n}-\theta_{0}\right)  \,;\,\left\{  \rho_{\mathrm{in}}%
^{n},\mathfrak{F}^{n},M^{n}\right\}  \text{ with (\ref{asym-unbiased}%
)}\right\}  \geq\left(  G_{\Phi}^{\min,a}\left(  \Delta\right)  \right)
^{-1}. \label{CR}%
\end{equation}

Indeed, one can show the identity in (\ref{CR}) is achievable, if $G_{\Phi
}^{\min,a}\left(  \Delta\right)  <\infty$ and some regularity conditions are
satisfied\thinspace\cite{Matsumoto:akahira}.

\subsection{On `Heisenberg rate'}

If $\left\{  \Phi_{\theta}\right\}  $ are unitary operations, due to
(\ref{unitary-infty}), (\ref{CR}) does not give any information on the
efficiency. Indeed, a number of literatures show that $\mathbb{E}_{\theta_{0}%
}\left(  \hat{\theta}^{n}-\theta_{0}\right)  \,=O\left(  \frac{1}{n^{2}%
}\right)  $ (Heisenberg rate), and some refers to application to metrology.
However, the efficiency of the optimal estimator is very weak against the
noise in the operations, as some authors have pointed out in some physical models.

Combination of \ (\ref{noisy-finite}) and (\ref{CR}) shows a general result
\cite{Matsumoto:akahira} :

\begin{theorem}
\label{th:O(1/n)}Suppose $\dim\mathcal{H}<\infty$ and that there is a
$\varepsilon_{\theta}$ such that $\Phi_{\theta}+\varepsilon_{\theta}\left(
\mathrm{d}\Phi_{\theta}/\mathrm{d}\theta\right)  $ and $\Phi-\varepsilon
_{\theta}\left(  \mathrm{d}\Phi_{\theta}/\mathrm{d}\theta\right)  $ are
completely positive. Then if $\,\left\{  \rho_{\mathrm{in}}^{n},\mathfrak{F}%
^{n},M^{n}\right\}  $ satisfies (\ref{asym-unbiased}),
\[
\mathbb{E}_{\theta}\left(  \hat{\theta}^{n}-\theta\right)  \,=O\left(
\frac{1}{n}\right)  ,\,\forall\theta\text{.}%
\]

\end{theorem}

Therefore, whatever the noise it is, however small it is, Heisenberg rate
collapses. Note Theorem\thinspace\ can be easily extended to the case that
$\theta$ is multi-dimensional. Obtaining estimate $\hat{\theta}^{n}$ of
multi-dimensional parameter $\theta$ satisfying (\ref{asym-unbiased}) for each
components. Then, its first component $\hat{\theta}^{n,1}$ is a estimate of
scalar parameter $\theta^{1}$ with (\ref{asym-unbiased}). Hence, due to
Theorem\thinspace\ref{th:O(1/n)}, we have
\[
\mathbb{E}_{\theta}\left\Vert \hat{\theta}^{n,1}-\theta^{1}\right\Vert
^{2}\geq\mathbb{E}_{\theta}\left(  \hat{\theta}^{n,1}-\theta^{1}\right)
\,=O\left(  \frac{1}{n}\right)  .
\]

\section{Asymptotic theory with approximation}

\subsection{Motivations}

Axiom (N) is justified because this is consequence of the rest of the axioms and

\begin{description}
\item[(N')] $g_{q}\left(  \delta^{\prime}\right)  =1$, if $\left\{
q,\delta^{\prime}\right\}  =\left\{  \mathrm{N}\left(  0,1\right)
,\delta\mathrm{N}\left(  0,1\right)  \right\}  $

\item[(C1)] (\textit{parallel weak asymptotic continuity)} If $\left\Vert
\Phi^{n}-\Phi^{\otimes n}\right\Vert _{\mathrm{cb}}\rightarrow0$ and $\frac
{1}{\sqrt{n}}\left\Vert \Delta^{n}-\Delta^{\left(  n\right)  }\right\Vert
_{\mathrm{cb}}\rightarrow0$ then
\[
\varliminf_{n\rightarrow\infty}\frac{1}{n}\left(  G_{\Phi^{n}}\left(
\Delta^{n}\right)  -G_{\Phi^{\otimes n}}\left(  \Delta^{\left(  n\right)
}\right)  \right)  \geq0.
\]

\end{description}

The proof uses asymptotic tangent simulation \cite{Matsumoto:2010-2}, which
simulates $\{p^{\otimes n},\delta^{\left(  n\right)  }\}$ by Gaussian shift
$\left\{  q,\delta^{\prime}\right\}  =\left\{  \mathrm{N}\left(  0,1\right)
,\delta\mathrm{N}\left(  0,1\right)  \right\}  $ only approximately. Hence,
for the sake of coherency, it is preferable to build a theory based on
asymptotic tangent simulation.

\subsection{ Asymptotic tangent simulation (parallel) and $\tilde{G}^{p,\max}%
$}

An \textit{asymptotic parallel classical tangent simulation }is a
sequence\textit{ }$\left\{  q^{n},\delta^{\prime n},\Lambda^{n}\right\}
_{n=1}^{\infty}$ of $q^{n}\in$ $\mathcal{P}_{\mathrm{pr}}$ ($\sigma^{n}%
\in\mathcal{S}_{\mathrm{pr}}$), $\delta^{\prime n}\in\mathcal{T}_{q}\left(
\mathcal{P}_{\mathrm{pr}}\right)  $ and a CPTP map $\Lambda^{n}$, such that%

\begin{equation}
\lim_{n\rightarrow\infty}\left\Vert \Phi^{\otimes n}-\Lambda^{n}\left(
\mathbf{I}\otimes q^{n}\right)  \right\Vert _{\mathrm{cb}}=0,
\label{asym-simulation-tangent-1}%
\end{equation}

and
\begin{equation}
\lim_{n\rightarrow\infty}\frac{1}{\sqrt{n}}\left\Vert \Delta^{\left(
n\right)  }-\Lambda^{n}\circ\left(  \mathbf{I}\otimes\delta^{n}\right)
\right\Vert _{\mathrm{cb}}=0, \label{asym-simulation-tangent-2}%
\end{equation}
where $\Delta\in\mathcal{T}_{\Phi}\left(  \mathcal{C}\right)  $, $\delta
\in\mathcal{T}_{q}\left(  \mathcal{P}_{\mathrm{pr}}\right)  $.

Based on this, we define the following quantity.%
\[
\tilde{G}_{\Phi}^{p,\max}\left(  \Delta\right)  :=\lim_{n\rightarrow\infty
}\frac{1}{n}\inf\left\{  J_{q^{n}}\left(  \delta^{\prime n}\right)  ;\left\{
q^{n},\delta^{\prime n}\right\}  \text{ is a Gaussian shift with
(\textit{\ref{asym-simulation-tangent-1}}),
(\textit{\ref{asym-simulation-tangent-2}})}\right\}  .
\]
(Here note that $\lim$ always exists and finite.)

Also, we define
\[
G_{\Phi}^{p,R}\left(  \Delta\right)  :=\lim_{n\rightarrow\infty}\frac{1}%
{n}\sup_{\rho}J_{\Phi^{\otimes n}\left(  \rho\right)  }^{R}\left(
\Delta^{\left(  n\right)  }\left(  \rho\right)  \right)  .
\]

\begin{theorem}
\label{th:G<G<G-p-2}(M), (E), (A1), (C1) and (N') implies that
\[
G_{\Phi}^{p,\min}\left(  \Delta\right)  \leq G_{\Phi}\left(  \Delta\right)
\leq\tilde{G}_{\Phi}^{p,\max}\left(  \Delta\right)  .
\]

\end{theorem}

\begin{proof}%
\begin{align*}
0  &  \leq\varliminf_{n\rightarrow\infty}\frac{1}{n}\left(  G_{\Lambda
^{n}\left(  \mathbf{I}\otimes q^{n}\right)  }\left(  \Lambda^{n}\left(
\mathbf{I}\otimes\delta^{n}\right)  \right)  -G_{\Phi^{\otimes n}}\left(
\Delta^{\left(  n\right)  }\right)  \right) \\
&  \leq\varliminf_{n\rightarrow\infty}\frac{1}{n}\left(  G_{\mathbf{I}\otimes
q^{n}}\left(  \mathbf{I}\otimes\delta^{n}\right)  -G_{\Phi^{\otimes n}}\left(
\Delta^{\left(  n\right)  }\right)  \right) \\
&  =\varliminf_{n\rightarrow\infty}\frac{1}{n}g_{q^{n}}\left(  \delta
^{n}\right)  -G_{\Phi}\left(  \Delta\right)  =\varliminf_{n\rightarrow\infty
}\frac{1}{n}J_{q^{n}}\left(  \delta^{n}\right)  -G_{\Phi}\left(
\Delta\right)  .
\end{align*}

\end{proof}

\begin{theorem}
\label{th:asym-cont-Gp}If $\dim\mathcal{H}_{\mathrm{in}}\,<\infty$ and
$\dim\mathcal{H}_{\mathrm{out}}\,<\infty$, $G^{p,\min}$ and $G^{p,R}$ satisfy
(M), (E), (A1), (C1) and (N').
\end{theorem}

\begin{proof}
That $G^{p,\min}$ and $G^{p,R}$ satisfy (M), (E), (A1), and (N') is trivial.
Hence, we prove (C1). Choose $\rho_{l,\varepsilon}$ so that
\[
\frac{1}{l}J_{\Phi^{\otimes l}\left(  \rho_{l,\varepsilon}\right)  }%
^{S}\left(  \Delta^{\left(  l\right)  }\left(  \rho_{l,\varepsilon}\right)
\right)  \geq\frac{1}{l}G_{\Phi^{\otimes l}}^{\min}\left(  \Delta^{\left(
l\right)  }\right)  -\varepsilon.
\]
Also, let%
\begin{align*}
\sigma_{l,\varepsilon}  &  :=\Phi^{\otimes l}\left(  \rho_{l,\varepsilon
}\right)  ,\\
\delta_{l,\varepsilon}  &  :=\Delta^{\left(  l\right)  }\left(  \rho
_{l,\varepsilon}\right)  ,
\end{align*}
and%
\begin{align*}
\sigma_{l,m,\varepsilon}^{\prime}  &  :=\left(  \Phi^{\otimes m}+\Psi
_{m}\right)  \left(  \rho_{l,\varepsilon}^{\otimes\left(  m/l\right)
}\right)  ,\\
\delta_{l,m,\varepsilon}^{\prime}  &  :=\left(  \Delta^{\left(  m\right)
}+D_{m}\right)  \left(  \rho_{m,\varepsilon}^{\otimes\left(  m/l\right)
}\right)  .
\end{align*}
Then
\begin{align*}
\left\Vert \sigma_{l,m,\varepsilon}^{\prime}-\sigma_{l,\varepsilon}^{\otimes
m}\right\Vert _{1}  &  \leq\left\Vert \Psi_{m}\right\Vert _{\mathrm{cb}},\\
\left\Vert \delta_{l,m,\varepsilon}^{\prime}-\delta_{l,\varepsilon}^{\left(
m\right)  }\right\Vert _{1}  &  \leq\left\Vert D_{m}\right\Vert _{\mathrm{cb}%
}.
\end{align*}
Observe by Schwartz's inequality,
\[
J_{\sigma_{l,\varepsilon}^{\otimes\left(  m/l\right)  }}^{S}\left(
\delta_{l,\varepsilon}^{\left(  m/l\right)  }\right)  \geq\frac{\left\vert
\mathrm{tr}\,\delta_{l,\varepsilon}^{\left(  m/l\right)  }X\right\vert ^{2}%
}{\mathrm{tr}\,\sigma_{l,\varepsilon}^{\otimes\left(  m/l\right)  }X^{2}}%
\]
and the equality is achieved by
\[
X_{m,l,\varepsilon}=\frac{l}{mJ_{l,\varepsilon}^{S}}L_{\sigma_{l,\varepsilon
}^{\otimes\left(  m/l\right)  }}^{S}\left(  \delta_{l,\varepsilon}^{\left(
m/l\right)  }\right)  ,
\]
where $J_{l,\varepsilon}^{S}=J_{\sigma_{l,\varepsilon}}^{S}\left(
\delta_{l,\varepsilon}\right)  $. Let $X_{m,l,\varepsilon}=\int xE\left(
d\mathrm{d}\,x\right)  $ be the spectral decomposition, and define
$P_{a}:=\int_{x\leq a}E\left(  d\mathrm{d}\,x\right)  .$ Then,
\begin{align*}
\frac{1}{m}J_{\sigma_{l,m,\varepsilon}^{\prime}}^{S}\left(  \delta
_{l,m,\varepsilon}^{\prime}\right)   &  \geq\frac{\left\vert \mathrm{tr}%
\,\delta_{l,m,\varepsilon}^{\prime}X_{m,l,\varepsilon}P_{a}\right\vert ^{2}%
}{m\mathrm{tr}\,\sigma_{l,m,\varepsilon}^{\prime}\left(  X_{m,l,\varepsilon
}P_{a}\right)  ^{2}}=\frac{\left\vert \frac{1}{\sqrt{m}}\mathrm{tr}%
\,\delta_{l,m,\varepsilon}^{\prime}X_{m,l,\varepsilon}P_{a}\right\vert ^{2}%
}{\mathrm{tr}\,\sigma_{l,m,\varepsilon}^{\prime}\left(  X_{m,l,\varepsilon
}P_{a}\right)  ^{2}}\\
&  \geq\frac{\left\vert \frac{1}{\sqrt{m}}\mathrm{tr}\,\delta_{l,\varepsilon
}^{\left(  m/l\right)  }X_{m,l,\varepsilon}P_{a}\right\vert ^{2}}%
{\mathrm{tr}\,\sigma_{l,\varepsilon}^{\otimes\left(  m/l\right)  }\left(
X_{m,l,\varepsilon}P_{a}\right)  ^{2}}+O\left(  \left\Vert \Psi_{m}\right\Vert
_{\mathrm{cb}}\right)  +O\left(  \frac{1}{\sqrt{m}}\left\Vert D_{m}\right\Vert
_{\mathrm{cb}}\right) \\
&  =\frac{\left\vert \frac{\sqrt{m}}{l}J_{l,\varepsilon}^{S}\mathrm{tr}%
\sigma_{l,\varepsilon}^{\otimes\left(  m/l\right)  }\left(  X_{m,l,\varepsilon
}\right)  ^{2}P_{a}\right\vert ^{2}}{\mathrm{tr}\,\sigma_{l,\varepsilon
}^{\otimes\left(  m/l\right)  }\left(  X_{m,l,\varepsilon}P_{a}\right)  ^{2}%
}+O\left(  \left\Vert \Psi_{m}\right\Vert _{\mathrm{cb}}\right)  +O\left(
\frac{1}{\sqrt{m}}\left\Vert D_{m}\right\Vert _{\mathrm{cb}}\right) \\
&  =\frac{m}{l^{2}}\left(  J_{l,\varepsilon}^{S}\right)  ^{2}\mathrm{tr}%
\sigma_{l,\varepsilon}^{\otimes\left(  m/l\right)  }\left(  X_{m,l,\varepsilon
}\right)  ^{2}P_{a}+O\left(  \left\Vert \Psi_{m}\right\Vert _{\mathrm{cb}%
}\right)  +O\left(  \frac{1}{\sqrt{m}}\left\Vert D_{m}\right\Vert
_{\mathrm{cb}}\right)  .
\end{align*}
On the other hand,
\begin{align*}
&  \left\vert \frac{m}{l^{2}}\left(  J_{l,\varepsilon}^{S}\right)
^{2}\mathrm{tr}\sigma_{l,\varepsilon}^{\otimes\left(  m/l\right)  }\left(
X_{m,l,\varepsilon}\right)  ^{2}-\frac{m}{l^{2}}\left(  J_{l,\varepsilon}%
^{S}\right)  ^{2}\mathrm{tr}\sigma_{l,\varepsilon}^{\otimes m}\left(
X_{m,l,\varepsilon}\right)  ^{2}P_{a}\right\vert \\
&  =\left\vert \frac{m}{l^{2}}\left(  J_{l,\varepsilon}^{S}\right)
^{2}\mathrm{tr}\sigma_{l,\varepsilon}^{\otimes\left(  m/l\right)  }\left(
X_{m,l,\varepsilon}\right)  ^{2}\left(  \mathbf{1}-P_{a}\right)  \right\vert
\\
&  \leq\frac{1}{a^{2}}\left\vert \frac{m}{l^{2}}\left(  J_{l,\varepsilon}%
^{S}\right)  ^{2}\mathrm{tr}\sigma_{l,\varepsilon}^{\otimes\left(  m/l\right)
}\left(  X_{m,l,\varepsilon}\right)  ^{4}\right\vert \\
&  =\frac{l^{2}}{a^{2}\left(  J_{l,\varepsilon}^{S}\right)  ^{2}m^{3}}\left\{
\frac{m}{l}\left(  \frac{m}{l}-1\right)  J_{l,\varepsilon}^{S}+\frac{m}%
{l}\mathrm{tr}\sigma_{l,\varepsilon}\left(  L_{\sigma_{l,\varepsilon}}%
^{S}\left(  \delta_{l,\varepsilon}\right)  \right)  ^{4}\right\} \\
&  \leq.\frac{1}{a^{2}\left(  J_{l,\varepsilon}^{S}\right)  ^{2}m}\left\{
J_{l,\varepsilon}^{S}+\mathrm{tr}\sigma_{l,\varepsilon}\left(  L_{\sigma
_{l,\varepsilon}}^{S}\left(  \delta_{l,\varepsilon}\right)  \right)
^{4}\right\}
\end{align*}
Therefore,
\begin{align*}
&  \frac{1}{m}J_{\sigma_{l,m,\varepsilon}^{\prime}}^{S}\left(  \delta
_{l,m,\varepsilon}^{\prime}\right) \\
&  \geq\frac{m}{l^{2}}\left(  J_{l,\varepsilon}^{S}\right)  ^{2}%
\mathrm{tr}\sigma_{l,\varepsilon}^{\otimes m}\left(  X_{m,l,\varepsilon
}\right)  ^{2}\\
&  +O\left(  \left\Vert \Psi_{m}\right\Vert _{\mathrm{cb}}\right)  +O\left(
\frac{1}{\sqrt{m}}\left\Vert D_{m}\right\Vert _{\mathrm{cb}}\right)  -\frac
{1}{a^{2}\left(  J_{l,\varepsilon}^{S}\right)  ^{2}m}\left\{  J_{l,\varepsilon
}^{S}+\mathrm{tr}\sigma_{l,\varepsilon}\left(  L_{\sigma_{l,\varepsilon}}%
^{S}\left(  \delta_{l,\varepsilon}\right)  \right)  ^{4}\right\} \\
&  =\frac{1}{l}J_{l,\varepsilon}^{S}\\
&  +O\left(  \left\Vert \Psi_{m}\right\Vert _{\mathrm{cb}}\right)  +O\left(
\frac{1}{\sqrt{m}}\left\Vert D_{m}\right\Vert _{\mathrm{cb}}\right)  -\frac
{1}{a^{2}\left(  J_{l,\varepsilon}^{S}\right)  ^{2}m}\left\{  J_{l,\varepsilon
}^{S}+\mathrm{tr}\sigma_{l,\varepsilon}\left(  L_{\sigma_{l,\varepsilon}}%
^{S}\left(  \delta_{l,\varepsilon}\right)  \right)  ^{4}\right\}  .
\end{align*}
Hence,
\begin{align*}
&  \frac{1}{m}G_{\Phi^{\otimes m}+\Psi_{m}}^{p,\min}\left(  \Delta^{\left(
m\right)  }+D_{m}\right) \\
&  =\lim_{n\rightarrow\infty}\frac{1}{mn}G_{\left(  \Phi^{\otimes m}+\Psi
_{m}\right)  ^{\otimes n}}^{\min}\left(  \left(  \Delta^{\left(  m\right)
}+D_{m}\right)  ^{\left(  n\right)  }\right)  \geq\frac{1}{m}G_{\Phi^{\otimes
m}+\Psi_{m}}^{\min}\left(  \Delta^{\left(  m\right)  }+D_{m}\right) \\
&  \geq\frac{1}{m}J_{\sigma_{l,m,\varepsilon}^{\prime}}^{S}\left(
\delta_{l,m,\varepsilon}^{\prime}\right) \\
&  \geq\frac{1}{l}G_{\Phi^{\otimes l}}^{\min}\left(  \Delta^{\left(  l\right)
}\right)  -\varepsilon\\
&  +O\left(  \left\Vert \Psi_{m}\right\Vert _{\mathrm{cb}}\right)  +O\left(
\frac{1}{\sqrt{m}}\left\Vert D_{m}\right\Vert _{\mathrm{cb}}\right)  -\frac
{1}{a^{2}\left(  J_{l,\varepsilon}^{S}\right)  ^{2}m}\left\{  J_{l,\varepsilon
}^{S}+\mathrm{tr}\sigma_{l,\varepsilon}\left(  L_{\sigma_{l,\varepsilon}}%
^{S}\left(  \delta_{l,\varepsilon}\right)  \right)  ^{4}\right\} \\
&  \rightarrow\frac{1}{l}G_{\Phi^{\otimes l}}^{\min}\left(  \Delta^{\left(
l\right)  }\right)  -\varepsilon\text{ }(m\rightarrow\infty)\\
&  \rightarrow G_{\Phi}^{p,\min}\left(  \Delta\right)  -\varepsilon
\quad(l\rightarrow\infty).
\end{align*}
Therefore,
\begin{align*}
&  \varliminf_{m\rightarrow\infty}\frac{1}{m}\left\{  G_{\Phi^{\otimes m}%
+\Psi_{m}}^{p,\min}\left(  \Delta^{\left(  m\right)  }+D_{m}\right)  -\frac
{1}{m}G_{\Phi^{\otimes m}}^{p,\min}\left(  \Delta^{\left(  m\right)  }\right)
\right\} \\
&  \geq G_{\Phi}^{p,\min}\left(  \Delta\right)  -\varepsilon-\varliminf
_{m\rightarrow\infty}\frac{1}{m}G_{\Phi^{\otimes m}}^{p,\min}\left(
\Delta^{\left(  m\right)  }\right) \\
&  =G_{\Phi}^{p,\min}\left(  \Delta\right)  -\varepsilon-G_{\Phi}^{p,\min
}\left(  \Delta\right)  =-\varepsilon
\end{align*}
Since $\varepsilon>0$ is arbitrary, we have (C1) for $G^{p,\min}$. (C1) for
$G^{p,R}$ is proved almost parallelly, utilizing the following consequence of
Schwartz's inequality:%
\[
J_{\rho}^{R}\left(  \delta\right)  \geq\frac{\left\vert \mathrm{tr}\,\delta
X\right\vert ^{2}}{\mathrm{tr}\,\rho X^{\dagger}X},
\]
where $X$ is an arbitrary operator.
\end{proof}

\begin{corollary}%
\[
G_{\Phi}^{p,\min}\left(  \Delta\right)  \leq G_{\Phi}^{p,R}\left(
\Delta\right)  \leq\tilde{G}_{\Phi}^{p,\max}\left(  \Delta\right)  .
\]

\end{corollary}

\subsection{Quantum states}

In \cite{Matsumoto:rld}, it had been essentially proved
\begin{align*}
G^{\min}  &  =G^{\min,p}=G^{\min,a}=J^{S},\\
G^{\max}  &  =G^{\max,p}=G^{\max,a}=J^{R}.
\end{align*}
It is not difficult to see
\[
J^{R}=G^{p,R}.
\]
Therefore, we have
\[
G^{p,R}=J^{R}\leq\tilde{G}^{\max,p}.
\]
On the other hand, if $\dim\mathcal{H<\infty}$, using \cite{Matsumoto:2010-2},
$J^{R}\geq\tilde{G}^{\max,p}$. Thus,
\[
G^{\max}=G^{\max,p}=G^{\max,a}=\tilde{G}^{\max,p}=J^{R}.
\]
The following theorem is a corollary of Theorem\thinspace\ref{th:asym-cont-Gp}.

\begin{theorem}
$J_{\rho}^{S}$ and $J_{\rho}^{R}$ satisfies (C), i.e.,
\[
\varliminf_{n\rightarrow\infty}\frac{1}{n}\left(  J_{\rho^{\prime n}}%
^{S}\left(  \delta^{\prime n}\right)  -J_{\rho^{\otimes n}}^{S}\left(
\delta^{\left(  n\right)  }\right)  \right)  \geq0,\quad\varliminf
_{n\rightarrow\infty}\frac{1}{n}\left(  J_{\rho^{\prime n}}^{R}\left(
\delta^{\prime n}\right)  -J_{\rho^{\otimes n}}^{R}\left(  \delta^{\left(
n\right)  }\right)  \right)  \geq0
\]
if $\left\Vert \rho^{\prime n}-\rho^{\otimes n}\right\Vert _{1}\rightarrow0$
and $\frac{1}{\sqrt{n}}\left\Vert \delta^{\prime n}-\delta^{\left(  n\right)
}\right\Vert _{1}\rightarrow0$.
\end{theorem}

\subsection{Classical channel}

By \cite{Matsumoto:2010-2}, we have the following :

\begin{theorem}
\label{th:finite-sim-prob}Suppose $p$ is a probability distribution and
$\delta$ is a signed measure over set with $k$-elements ($k<\infty$). Let
$J:=J_{p}\left(  \delta\right)  $, $\varepsilon>0$ and
\[
\left\{  q^{n},\delta^{\prime n}\right\}  :=\left\{  \mathrm{N}\left(
0,1\right)  ,\sqrt{n\left(  J+\varepsilon\right)  }\delta\mathrm{N}\left(
0,1\right)  \right\}  \equiv\left\{  \mathrm{N}\left(  0,1\right)
,\delta\mathrm{N}\left(  0,1\right)  \right\}  ^{\otimes n\left(
J+\varepsilon\right)  }.
\]
Then, we can compose an asymptotic parallel tangent simulation of $\left\{
p^{\otimes n},\delta^{\left(  n\right)  }\right\}  $ using $\left\{
q^{n},\delta^{\prime n}\right\}  $.
\end{theorem}

\begin{theorem}
\label{th:channel-sim-finite}$\Phi\left(  \cdot|x\right)  $ is a probability
distribution and $\Delta\left(  \cdot|x\right)  $ is a signed measure over set
with $k$-elements ($k<\infty$). Let us define $\left\{  q_{\varepsilon}%
^{n},\delta_{\varepsilon}^{n}\right\}  :=$ $\{\mathrm{N}\left(  0,1\right)
,\delta\mathrm{N}\left(  0,1\right)  \}^{\otimes n\left(  1+k\varepsilon
\right)  J}$ where $J=G_{\Phi}^{\min}\left(  \Delta\right)  =\max_{1\leq x\leq
k}J_{\Phi\left(  \cdot|x\right)  }\left(  \Delta\left(  \cdot|x\right)
\right)  $ and $\varepsilon$ is arbitrary positive number. Then, there is
$\Lambda^{n}$ such that
\begin{align*}
\left\Vert \Phi^{\otimes n}\left(  p\right)  -\Lambda^{n}\left(  p\otimes
q_{\varepsilon}^{n}\right)  \right\Vert _{\mathrm{cb}}  &  \leq\frac{k}%
{\sqrt{\varepsilon n}}\max_{x}C\left(  \left\{  \Phi\left(  \cdot|x\right)
,\Delta\left(  \cdot|x\right)  \right\}  \right)  ,\\
\frac{1}{\sqrt{n}}\left\Vert \Delta^{\left(  n\right)  }\left(  p\right)
-\Lambda^{n}\left(  p\otimes\delta_{\varepsilon}^{n}\right)  \right\Vert
_{\mathrm{cb}}  &  \leq\frac{k}{\sqrt{\varepsilon n}}\max_{x}C\left(  \left\{
\Phi\left(  \cdot|x\right)  ,\Delta\left(  \cdot|x\right)  \right\}  \right)
.
\end{align*}

\end{theorem}

\subsection{CQ channel}

Let
\[
\Phi\left(  x\right)  =\rho_{x},\,\,\Delta\left(  x\right)  =\delta_{x}.
\]
Then, trivially%

\[
G_{\Phi}^{\min,p}\left(  \Delta\right)  =\max_{x}J_{\rho_{x}}^{S}\left(
\delta_{x}\right)  .
\]
Also,%

\[
\tilde{G}_{\Phi}^{\max,p}\left(  \Delta\right)  \geq G_{\Phi}^{R,p}\left(
\Delta\right)  =\max_{x}J_{\rho_{x}}^{R}\left(  \delta_{x}\right)  .
\]

In the sequel, we prove $\tilde{G}_{\Phi}^{\max,p}\left(  \Delta\right)  \leq
G_{\Phi}^{R,p}\left(  \Delta\right)  =\max_{x}J_{\rho_{x}}^{R}\left(
\delta_{x}\right)  $, if $\dim\mathcal{H}<\infty$. \ Let $\left\{
q_{x},\delta_{x}^{\prime}\right\}  $ be the optimal tangent simulation of
$\left\{  \rho_{x},\delta_{x}\right\}  $, or satisfy
\[
\rho_{x}=\sum_{y}q_{x}\left(  y\right)  \rho_{x,y}\,,\,\,\,\delta_{x}=\sum
_{y}\delta_{x}^{\prime}\left(  y\right)  \rho_{x,y}\,,
\]
and $J_{q_{x}}\left(  \delta_{x}^{\prime}\right)  =J_{\rho_{x}}^{R}\left(
\delta_{x}\right)  $ (see \cite{Matsumoto:rld}).

Denote $\max_{x}J_{\rho_{x}}^{R}\left(  \delta_{x}\right)  $ by $J$. Given a
input sequence $x^{n}=x_{1}x_{2}\cdots x_{n}$, denote by $n_{x}$ the tmes of
$x_{i}=x$ in the sequence $x^{n}$. Suppose $n_{x}\geq\varepsilon n$. Then, we
use $\{\mathrm{N}\left(  0,1\right)  ,\delta\mathrm{N}\left(  0,1\right)
\}^{\otimes n_{x}J}$ for simulation of $\left\{  \rho_{x}^{\otimes n_{x}%
},\delta_{x}^{\left(  n_{x}\right)  }\right\}  $. On the other hand, if
$n_{x}<\varepsilon n$, we first fabricate $\left\{  \rho_{x}^{\otimes
\varepsilon n},\delta_{x}^{\left(  \varepsilon n\right)  }\right\}  $
consuming $\{\mathrm{N}\left(  0,1\right)  ,\delta\mathrm{N}\left(
0,1\right)  \}^{\otimes n_{x}J}$, and takes partial trace. In both case, by
Theorem\thinspace\ref{th:finite-sim-prob}, the error of simulation vanishes as
$n\rightarrow\infty$. We do this for all $x=1,\cdots,k$. As a whole, we used
$\left\{  \mathrm{N}\left(  0,1\right)  ,\delta\mathrm{N}\left(  0,1\right)
\right\}  ^{\otimes\left(  n+\varepsilon kn\right)  J}$ to simulate
$\bigotimes_{\kappa=1}^{n}\left\{  \rho_{x_{\kappa}}^{\otimes n_{x_{\kappa}}%
},\delta_{x_{\kappa}}^{\left(  n_{x_{\kappa}}\right)  }\right\}  $. \ Since
$\varepsilon>0$ is arbitrary, we have
\[
\tilde{G}_{\Phi}^{\max,p}\left(  \Delta\right)  =G_{\Phi}^{R,p}\left(
\Delta\right)  =\max_{x}J_{\rho_{x}}^{R}\left(  \delta_{x}\right)  .
\]

\begin{conjecture}
For a QC channel, $G_{\Phi}^{R,p}\left(  \Delta\right)  =\tilde{G}_{\Phi
}^{\max,p}\left(  \Delta\right)  \lneqq G_{\Phi}^{\max,p}\left(
\Delta\right)  $.
\end{conjecture}

\end{document}